\documentclass[11pt,reqno,letterpaper]{amsart}

\usepackage{booktabs}
\usepackage{color}
\usepackage[colorlinks=true, allcolors=blue,backref=page]{hyperref}
\usepackage{amsmath, amssymb, amsthm}
\usepackage{mathrsfs}
\usepackage{mathtools}
\usepackage[noabbrev,capitalize,nameinlink]{cleveref}
\crefname{equation}{}{}
\usepackage[noadjust]{cite}
\usepackage{graphics}
\usepackage{pifont}
\usepackage{tikz}
\usepackage{bbm}
\usepackage{thm-restate}
\usepackage[T1]{fontenc}

\DeclareMathOperator*{\Cor}{\text{Cor}}
\DeclareMathOperator*{\op}{\textrm{op}}
\usetikzlibrary{arrows.meta}

\usepackage{environ}
\usepackage{framed}
\usepackage{url}
\usepackage[linesnumbered,ruled,vlined]{algorithm2e}
\usepackage[noend]{algpseudocode}
\usepackage[labelfont=bf]{caption}
\usepackage[framemethod=tikz]{mdframed}
\usepackage{appendix}
\usepackage{graphicx}
\usepackage[textsize=tiny]{todonotes}
\usepackage{tcolorbox}
\allowdisplaybreaks[1]
\usepackage{enumerate}
\usepackage{stmaryrd}

\usepackage[margin=1in]{geometry}

\usepackage[shortlabels]{enumitem}
\crefformat{enumi}{#2#1#3}
\crefrangeformat{enumi}{#3#1#4 to~#5#2#6}
\crefmultiformat{enumi}{#2#1#3}
{ and~#2#1#3}{, #2#1#3}{ and~#2#1#3}

\DeclareSymbolFont{symbolsC}{U}{pxsyc}{m}{n}
\SetSymbolFont{symbolsC}{bold}{U}{pxsyc}{bx}{n}
\DeclareFontSubstitution{U}{pxsyc}{m}{n}
\DeclareMathSymbol{\medcircle}{\mathbin}{symbolsC}{7}

\crefname{algocf}{Algorithm}{Algorithms}

\crefname{equation}{}{} %remove ``Equation''
\AtBeginEnvironment{appendices}{\crefalias{section}{appendix}} %appendices

\usepackage[color,final]{showkeys} %add in 'final' into parameter to remove showkeys

\colorlet{refkey}{orange!20}
\colorlet{labelkey}{blue!30}

\crefname{algocf}{Algorithm}{Algorithms}

\renewcommand{\P}{\mathbb{P}}

\newcommand{\Var}{\operatorname{Var}}
\newcommand{\Cov}{\operatorname{Cov}}

\newcommand{\R}{\mathbb{R}}

\newcommand{\cI}{\mathcal{I}}
\newcommand{\supp}{\text{supp}}
\newcommand{\set}[1]{\{#1 \}}

\numberwithin{equation}{section}
\newtheorem{theorem}{Theorem}[section]
\newtheorem{proposition}[theorem]{Proposition}
\newtheorem{lemma}[theorem]{Lemma}

\crefname{claim}{Claim}{Claims}

\newtheorem*{question*}{Question}
\newtheorem{fact}[theorem]{Fact}

\theoremstyle{definition}
\newtheorem{definition}[theorem]{Definition}

\newtheorem*{definition*}{Definition}

\newtheorem{remark}[theorem]{Remark}



\newcommand{\mb}{\mathbb}

\newcommand{\mc}{\mathcal}

\newcommand{\on}{\operatorname}

\newcommand{\wt}{\widetilde}

\DeclareMathOperator{\Ent}{Ent}
\newcommand{\KL}{\operatorname{KL}}	
\newcommand{\E}{\mathbb{E}}	

\newcommand{\eps}{\varepsilon}

\let\originalleft\left
\let\originalright\right
\renewcommand{\left}{\mathopen{}\mathclose\bgroup\originalleft}
\renewcommand{\right}{\aftergroup\egroup\originalright}

\allowdisplaybreaks

\newcommand{\ignore}[1]{}

\title{Entropic independence via sparse localization}
\author[A1]{Vishesh Jain}
\address{Department of Mathematics, Statistics, and Computer Science, University of Illinois Chicago, Chicago, IL, 60607 USA}
\email{visheshj@uic.edu}

\author[A2]{Huy Tuan Pham}
\address{Department of Mathematics, California Institute of Technology, Pasadena, CA 91125 USA}
\email{htpham@caltech.edu}

\author[A3]{Thuy-Duong Vuong}
\address{Department of Computer Science and Engineering, University of California, San Diego, La Jolla, CA 92093 USA}
\email{thvuong@ucsd.edu}

\begin{document}

\begin{abstract}
Entropic independence is a structural property of measures that underlies modern proofs of functional inequalities, notably (modified) log-Sobolev inequalities, via ``annealing'' or local-to-global schemes.  Existing sufficient criteria for entropic independence typically require spectral independence and/or uniform bounds on marginals under \emph{all} pinnings, which can fail in natural canonical-ensemble models even when strong mixing properties are expected.

We introduce \emph{sparse localization}: a restricted localization framework, in the spirit of Chen--Eldan, in which one assumes $\ell_2$-independence only for a sparse family of pinnings (those fixing at most $cn$ coordinates for any $c > 0$), yet still deduces quadratic entropic stability and entropic independence with an explicit multiplicative loss of order $c^{-1}$. As an application, we give a rigorous proof of approximate conservation of entropy for the uniform distribution on independent sets of a given size in bounded degree graphs. 

\end{abstract}

\maketitle
 
%\thispagestyle{empty}
%\newpage 
%\setcounter{page}{1}

\section{Introduction}
\label{sec:intro}

Functional inequalities such as (modified) log-Sobolev inequalities (mLSIs) are a standard route to sharp
mixing time bounds, concentration inequalities, and quantitative stability for high-dimensional discrete
distributions.
A widely-used modern strategy for proving mLSIs is an {annealing} (or {interpolation}) approach:
one constructs a path of intermediate measures
\[
\nu=\nu^{(0)} \rightsquigarrow \nu^{(1)} \rightsquigarrow \cdots \rightsquigarrow \nu^{(T)},
\]
and proves an mLSI for $\nu$ in two conceptually distinct steps:
\begin{enumerate}[label=(\roman*),leftmargin=2.2em]
    \item \emph{Reduction to an easy regime.} Show that entropy does not decay too quickly along the path,
    typically via an entropy-factorization / entropic contraction estimate.
    \item \emph{Handling the easy regime.} Prove an mLSI for $\nu^{(T)}$,
    using model-specific arguments in a tractable regime.
\end{enumerate}
Entropic independence, which is the focus of this paper, is concerned with step~(i). 

We focus on distributions on the $k$-slice $\binom{[n]}{k}$, where $[n]:=\{1,2,\dots,n\}$.
For a probability distribution $\mu$ on $\binom{[n]}{k}$, define its normalized one-site marginals
$q_\mu$ as a distribution on $[n]$:
\[
q_{\mu}(i) \;=\; \frac{\P_{S\sim \mu}[i \in S]}{k}
\qquad(i\in[n]).
\]

\begin{definition}[Entropic independence]
\label{def:entropic-independence-intro}
A probability distribution $\nu$ on $\binom{[n]}{k}$ is \emph{$C$-entropically independent} if for every
$\mu$ on $\binom{[n]}{k}$ with $\mu\ll \nu$,
\[
\KL(q_{\mu} \| q_{\nu}) \;\le\; \frac{C}{k}\,\KL(\mu \| \nu).
\]
\end{definition}

Entropic independence has been used repeatedly as a modular hypothesis for deriving mLSIs for natural Markov chains on combinatorial state spaces; see, e.g.,
\cite{ChenLiuVigoda2021, anari2021entropic, anari2021entropic2} and references therein.

\subsection{Existing criteria for entropic independence}

There are, by now, several powerful general frameworks that imply entropic independence from
correlation or influence information. We highlight three that are most relevant to this paper. We refer the reader to \cref{sec:l2-indep} for various definitions. 

\begin{enumerate}
\item Chen--Liu--Vigoda \cite{ChenLiuVigoda2021} show that {spectral independence} together with suitable {marginal
boundedness} hypotheses yields entropy factorization and mLSIs for a broad class of Markov chains.
A salient feature is that these hypotheses are formulated {uniformly over pinnings/conditionings} of the
base measure; this uniformity is crucial in the inductive/local-to-global machinery underlying their approach.

\item A different line of work by Anari et.~al \cite{anari2021entropic,anari2021entropic2} derives entropic
independence from strong global structure such as fractional log-concavity (which is equivalent to spectral independence under all tilts) and related
high-dimensional expansion properties of generating polynomials. These conditions are extremely powerful, but do not hold as widely. 

\item Chen--Eldan \cite{chen2022localization} introduce a martingale framework which implies entropic independence
under an appropriate spectral independence hypothesis that is again formulated {uniformly over all pinnings}, together
with additional regularity/tameness conditions on marginals. We also refer the reader to independent and concurrent work of Chen, Feng, Yin, and Zhang \cite{Chen2022OptimalMF}, which proves optimal mLSIs for general anti-ferromagnetic models. 
\end{enumerate}

A common thread is that all these criteria require controlling either (i) spectral independence, or typically, (ii) both spectral independence and marginal behavior,
under a very broad family of conditionings (in fact, all conditionings). Unfortunately, in many applications (such as those involving conservative dynamics, \cite{JMPV, kuchukova2025fast}, see also \cref{sec:application}), the ``all pinnings'' requirement in these criteria is too stringent. This mismatch is especially acute in so-called annealing arguments (i.e.~reduction to the easy regime): this step typically explores only a family of intermediate measures that retain many of the properties of the original measure, whereas ``full''
localization may generate highly adversarial conditionings with extreme marginals or qualitatively different
behavior.

\subsection{Our contribution}

The main conceptual message of this paper is that one does {not} need uniform control under all pinnings.
Instead, we introduce a {sparse localization} framework: we only require independence bounds for
pinnings that fix at most a $c$-fraction of the coordinates, for any $c > 0$, and we derive entropic independence with an
explicit loss of a factor $c^{-1}$ in constants.

Sparse localization has two key advantages:
\begin{enumerate}[label=(\roman*),leftmargin=2.2em]
    \item It matches the {restricted conditioning} landscape in applications such as \cite{JMPV}, where
    independence estimates are available along the relevant localization path but not uniformly
    over all pinnings.
    \item Sparse pinnings produce measures {close} to the original one: by forbidding near-total pinning,
    sparse localization avoids creating highly degenerate intermediate distributions, in contrast to full
    localization schemes that may explore a much larger and less stable family of conditionings.
\end{enumerate}

Our main theorem shows that marginal boundedness of the base measure together with sparse $\ell_2$-independence implies entropic
independence. Informally, sparse $\ell_2$-independence requires a uniform operator norm bound on the
{influence matrix} under all pinnings that touch at most $cn$ coordinates (see
\cref{sec:l2-indep} for the precise definition).

\begin{theorem}[Entropic independence]
\label{thm:entropic-independence}
Let $\nu$ be a probability measure on $\binom{[n]}{k}$. Suppose that:
\begin{enumerate}[label=(\alph*),leftmargin=2.2em]
\item \emph{(Marginal lower bound)} for every $i\in[n]$, either $\P_{S\sim\nu}[i\in S]=0$ or
$\P_{S\sim\nu}[i\in S]\ge b$ for some $b\in(0,1)$; and
\item \emph{(Sparse $\ell_2$-independence)} the pushforward $\widetilde{\nu}$ on $\{\pm1\}^n$ is
$\alpha$-$\ell_2$-independent restricted to $\mathrm{sign}(\mathrm{Sparse}_c)$, i.e.~for every feasible $u \in \mathrm{sign}(\mathrm{Sparse}_c)$,
\[\|\Psi(\mathcal{R}_u \wt{\nu}) \|_{\mathrm{op}} \leq \alpha.\]
\end{enumerate}
Then $\nu$ is $C$-entropically independent with $C=2\alpha/(bc)$.
\end{theorem}

\begin{remark}
When $c=1$, sparse localization becomes full localization and \cref{thm:main-quadratic-stability} recovers the Chen--Eldan quadratic
stability paradigm \cite[Theorem~42]{chen2022localization}. The point of our framework is that it remains effective when
independence or marginal boundedness is available only for a very sparse family of pinnings, at the price of the explicit factor
$c^{-1}$. The power of this comes from the fact that for very sparse pinnings $u$, $\mathcal{R}_u \nu$ often has quantitatively similar independence and marginal boundedness properties as $\nu$. 
\end{remark}

\begin{remark}
    Our assumptions are phrased in terms of $\ell_2$-independence i.e.~operator norm bounds on the influence matrix, as opposed to spectral norm bounds. We discuss the relation between these notions in \cref{sec:l2-indep}. 
\end{remark}

Our entropic independence theorem is obtained via a quadratic stability inequality on the cube, which may be
of independent interest.

\begin{theorem}[Quadratic stability from sparse $\ell_2$-independence]
\label{thm:main-quadratic-stability}
Let $\nu$ be a probability measure on $\{\pm 1\}^n$. Fix $c\in(0,1]$ and $\alpha>0$.
If $\nu$ is $\alpha$-$\ell_2$-independent restricted to $\mathrm{sign}(\mathrm{Sparse}_c)$, then for every
$\mu\ll \nu$,
\[
\|m(\mu)-m(\nu)\|_2^2 \;\le\; \frac{8\alpha}{c}\,\KL(\mu\|\nu).
\]
\end{theorem}

\paragraph{\bf Proof techniques}
The proof of \cref{thm:main-quadratic-stability} follows the Chen--Eldan localization blueprint but introduces
two new modular ingredients that enable {restricted} (sparse) pinning hypotheses.

\begin{itemize}
\item In \cref{sec:lipschitz} we prove a {restricted localization Lipschitz lemma}
(\cref{lem:lemma68}): $\ell_2$-independence over pinnings in $\mathrm{sign}(\mc V)$ implies Lipschitz control of
the mean map $v\mapsto m(\mc T_v\widetilde{\nu})$ for $v\in\mc V$ (and we apply this with
$\mc V=\mathrm{Sparse}_c$). This step is where we ensure the localization martingale only visits pinnings that
remain in the allowed family.

\item In \cref{sec:quad} we replace the ``full'' convex duality step of \cite{chen2022localization}
by a {sparse} convex duality argument: we apply the Donsker--Varadhan variational formula
(\cref{thm:DV}) but restrict test functions to {sparse linear} functionals, and we compute the resulting
restricted quadratic conjugate explicitly in terms of Ky Fan norms (\cref{lem:sparse-conjugate}). This is the
unique point where the factor $c^{-1}$ enters, through the inequality comparing the Ky Fan norm to $\|\cdot\|_2$
(\cref{lem:kyfan_vs_l2}). 
\end{itemize}

\cref{thm:entropic-independence} follows in a relatively simple manner from \cref{thm:main-quadratic-stability} using a standard inequality upper bounding the KL divergence by the $\chi^2$-divergence. This step only uses marginal boundedness for the base measure.\\ 

\paragraph{\bf Application: down-up walk on fixed size independent sets} In \cref{sec:application}, we discuss an application of \cref{thm:entropic-independence} to proving annealing (i.e.~reduction to the easy regime) for the uniform distribution on fixed size independent sets in bounded degree graphs. This step is crucial in the program of \cite{JMPV} for proving optimal mixing of the so-called down-up walk on fixed size independent sets throughout the algorithmic tractability regime; however, as was pointed out to us by Zongchen Chen and Mario Morellini, this step does not follow from any previous arguments in the literature (contrary to what is claimed in \cite{JMPV}), precisely due to the fact that previous methods for proving entropic independence require spectral independence and/or marginal boundedness for all pinnings of the base measure, which is not available here. Our argument in \cref{sec:application} fixes this critical gap.   

\subsection{Organization}

\cref{sec:prelims} records some basic definitions and notation. 
\cref{sec:l2-indep} defines influence matrices and sparse $\ell_2$-independence. \cref{sec:lipschitz} proves a sparse localization Lipschitz theorem. \cref{sec:quad} proves \cref{thm:main-quadratic-stability} and derives
\cref{thm:entropic-independence}. An application to the down-up walk on fixed-size independent sets is presented in \cref{sec:application}. 

\subsection{Acknowledgements} 
We thank Yuansi Chen, Zongchen Chen, and Mario Morellini for insightful discussions. V.J.~is partially supported by NSF CAREER DMS-2237646.

\section{Preliminaries}
\label{sec:prelims}
\subsection{Tilts, pinnings, and sparse vectors}
We work on the discrete cube $\{\pm 1\}^n$. For a probability measure $\nu$ on $\{\pm 1\}^n$, write
\[
m(\nu) := \mathbb{E}_\nu[X] \in [-1,1]^n,
\qquad
m_i(\nu) := \mathbb{E}_\nu[X_i],
\]
for the mean vector, where $X=(X_1,\dots,X_n)\sim \nu$. 

For $v\in\mathbb{R}^n$, define the {tilt} (or external field) operator $\mathcal{T}_v$ on measures by
\begin{equation}
\label{eq:TiltDef}
\frac{d(\mathcal{T}_v \nu)}{d\nu}(x)
=
\frac{e^{\langle v,x\rangle}}{\mathbb{E}_\nu[e^{\langle v,X\rangle}]},
\qquad x\in\{\pm 1\}^n.
\end{equation}

For $u\in\{-1,0,1\}^n$, define the {pinning} operator $\mathcal{R}_u$ by conditioning:
\begin{equation}
\label{eq:PinDef}
\mathcal{R}_u \nu := \nu\bigl(\,\cdot\mid X_i = u_i\text{ for all }i\text{ with }u_i\neq 0\,\bigr),
\end{equation}
whenever the conditioning event has positive $\nu$-probability. 

For $x \in \mb{R}$, we define the sign function by
\[\mathrm{sign}(x) = 
\begin{cases}
    1 \qquad & x > 0\\
    0 \qquad & x = 0\\
    -1 \qquad & x < 0.
\end{cases}
\]

For a set of vectors $\mc{V}\subseteq \mb{R}^n$, we define $\mathrm{sign}(\mathcal{V}) \subseteq \{-1,0,1\}^n$ to be the downward-closure of the vectors $\{\on{sign}(v): v \in \mc{V}\}$. Formally,
\begin{equation}
\mathrm{sign}(\mathcal{V})
:=
\Bigl\{s \in \{-1,0,1\}^{n}: \text{there exists } v \in \mc{V} \text{ such that } s_i \in \{0,\on{sign}(v_i)\}\ \forall i\in[n]\Bigr\}.
\end{equation}
\begin{remark}
    For our applications in \cref{sec:lipschitz}, it is important that we work with the downward-closure, as defined above, and not just with the set of all sign vectors themselves. 
\end{remark}

\begin{definition}[Sparse vectors]\label{def:sparse}
Fix $c\in(0,1]$.
Define
\[
\mathrm{Sparse}_c := \Bigl\{v\in\mathbb{R}^n:\ |\{i\in[n]: v_i\neq 0\}|\le \lceil cn \rceil\Bigr\}.
\]
Note that
\[
\mathrm{sign}(\mathrm{Sparse}_c)=\{-1,0,1\}^n\cap \mathrm{Sparse}_c.
\]
\end{definition}

\subsection{Entropy and divergences}\label{subsec:entropy}

For probability measures $\mu$ and $\pi$ on a common finite set $\Omega$, we write $\mu \ll \pi$ to mean that
$\mu$ is absolutely continuous with respect to $\pi$. In this case, $\frac{d\mu}{d\pi}(\omega)=\mu(\omega)/\pi(\omega)$.

\begin{definition}[Entropy functional]
Let $\pi$ be a probability measure on a finite set $\Omega$ and let $f:\Omega\to\mathbb{R}_{\ge 0}$ with $\mathbb{E}_\pi[f]>0$.
Define
\[
\Ent_\pi[f]
:=
\mathbb{E}_\pi\!\left[f\log\frac{f}{\mathbb{E}_\pi f}\right]
=
\mathbb{E}_\pi[f\log f]-\mathbb{E}_\pi[f]\log\mathbb{E}_\pi[f].
\]
\end{definition}

\begin{definition}[KL and $\chi^2$ divergence]
Let $\pi,\mu$ be probability measures on a finite set $\Omega$ with $\mu\ll \pi$. Define
\[
\KL(\mu \| \pi) = \mathbb{E}_{\mu}\left[\log \frac{d\mu}{d\pi}\right]
=\sum_{\omega\in\Omega} \mu(\omega)\log\frac{\mu(\omega)}{\pi(\omega)}.
\]
Also define
\[
\chi^2(\mu\| \pi) = \sum_{\omega \in \Omega} \frac{(\mu(\omega) - \pi(\omega))^2}{\pi(\omega)}.
\]
\end{definition}

The following standard fact follows using the elementary inequality $\log(1+x) \leq x$ for all $x \geq 0$.

\begin{fact}\label{fact:kl-chi2-inequality}
Let $\pi,\mu$ be probability measures on a finite set $\Omega$ with $\mu\ll \pi$. Then
\[
\KL(\mu\|\pi) \le \chi^2(\mu\|\pi).
\]
\end{fact}

\begin{definition}[$f$-tilt and KL identity]\label{def:f-tilt}
Given $\pi$ a probability measure on $\Omega$ and $f:\Omega\to\mathbb{R}_{\ge 0}$ with $\mathbb{E}_\pi f>0$,
define the tilted measure $\mu$ by
\[
\frac{d\mu}{d\pi}(x)=\frac{f(x)}{\mathbb{E}_\pi f}.
\]
Then
\begin{equation}\label{eq:EntKL}
\KL(\mu\|\pi)
=
\frac{\Ent_\pi[f]}{\mathbb{E}_\pi f}.
\end{equation}
\end{definition}

\begin{theorem}[Donsker--Varadhan variational formula]\label{thm:DV}
Let $\mu,\nu$ be probability measures on a finite set $\Omega$ with $\mu\ll\nu$. Then
\[
\KL(\mu\|\nu)
=\sup_{\varphi:\Omega\to\mathbb{R}}\left\{\mathbb{E}_\mu[\varphi]-\log\mathbb{E}_\nu[e^{\varphi}]\right\}.
\]
Moreover, if $g=\frac{d\mu}{d\nu}$, the supremum is attained by $\varphi^*=\log g + C$ for any constant $C\in\mathbb{R}$.
\end{theorem}

\subsection{From $\binom{[n]}{k}$ to the cube}\label{subsec:slice-embed}

We use the standard embedding $\binom{[n]}{k}\hookrightarrow\{0,1\}^n$ sending
$S\mapsto \mathbf{1}_S$, followed by the map $\{0,1\}^n\to\{\pm1\}^n$ given by $x\mapsto 2x-\mathbf{1}$.
For a measure $\mu$ on $\binom{[n]}{k}$, denote its pushforward on $\{\pm1\}^n$ by $\widetilde{\mu}$.
Then for each $i\in[n]$,
\[
m_i(\widetilde{\mu})=2\,\P_{S\sim \mu}[i\in S]-1.
\]
In particular,
\begin{equation}\label{eq:mean-vs-marg}
\|m(\widetilde{\mu})-m(\widetilde{\nu})\|_2^2
=
4\sum_{i=1}^n\left(\P_{S\sim \mu}[i\in S]-\P_{S\sim \nu}[i\in S]\right)^2.
\end{equation}

\section{Influence matrices and sparse $\ell_2$-independence}\label{sec:l2-indep}

\subsection{Correlation matrices and notions of independence}

For a measure $\nu$ on $\{\pm 1\}^n$, define the covariance matrix
\[
\mathrm{Cov}(\nu)_{ij}
:= \mathrm{Cov}_\nu(X_i,X_j)
= \mathbb{E}_\nu[X_iX_j]-\mathbb{E}_\nu[X_i]\mathbb{E}_\nu[X_j],
\]
and variances $\mathrm{Var}_\nu(X_i)=\mathrm{Cov}(\nu)_{ii}$. When $\mathrm{Var}_\nu(X_i)=0$, coordinate $i$ is deterministic under $\nu$.
In this case we interpret influence/correlation matrices as acting on the {active} coordinate set
$A(\nu):=\{i:\mathrm{Var}_\nu(X_i)>0\}$ (i.e.\ by restricting to the principal submatrix indexed by $A(\nu)$).
All operator norms below are taken on this active subspace.
Whenever $\mathrm{Var}_\nu(X_i)>0$ for all $i\in A(\nu)$, define the correlation matrix
\[
\mathrm{Cor}(\nu)
:=
\mathrm{diag}\!\bigl(\mathrm{Var}_\nu(X_i)\bigr)^{-1/2}
\,\mathrm{Cov}(\nu)\,
\mathrm{diag}\!\bigl(\mathrm{Var}_\nu(X_i)\bigr)^{-1/2},
\]
and the influence matrix
\[
\Psi(\nu)
:=
\mathrm{Cov}(\nu)\,
\mathrm{diag}\!\bigl(\mathrm{Var}_\nu(X_i)\bigr)^{-1}.
\]
(Here $\mathrm{diag}(\mathrm{Var}_\nu(X_i))$ denotes the diagonal matrix on $A(\nu)$ with entries
$\mathrm{Var}_\nu(X_i)$.)
We remark that when viewing $\nu$ as a distribution over $ 2^{[n]}$ induced by the map $ x\mapsto \set{i: x_i =1}$ then the entries of $ \Psi(\nu)$ can be rewritten as 
\[ \Psi(\nu)_{i,j} = \P_{S\sim \nu} [i \in S |j\in S] - \P_{S\sim \nu} [i \in S |j\not\in S]\]

\begin{definition}[limited $\ell_2$-independence]
\label{def:limited-ell_2-independence}
Let $\nu$ be a probability measure on $\{\pm 1\}^n$ and let $\mc{S} \subseteq \{-1,0,1\}^n$ with $0 \in \mc{S}$. For $\alpha > 0$, we say that $\nu$ is $\alpha$-$
\ell_2$-independent restricted to $\mc{S}$ if for any $u \in \mc{S}$,    
\[
\|\Psi(\mathcal{R}_u \nu)\|_{\mathrm{op}} \le \alpha.
\]
\end{definition}

Specializing to the case of $\mc{S} = \mathrm{sign}(\mathrm{Sparse}_c)$, gives the following. 

\begin{definition}[sparse $\ell_2$-independence]
Let $\nu$ be a probability measure on $\{\pm 1\}^n$. For $\alpha > 0$ and $c \in (0,1]$, we say that $\nu$ is $c$-sparse $\alpha$-$\ell_2$-independent if for any $u \in \{-1,0,1\}^n \cap \mathrm{Sparse}_c$,
\[
\|\Psi(\mathcal{R}_u \nu)\|_{\mathrm{op}} \le \alpha.
\]
\end{definition}

\paragraph{\bf $\ell_2$-independence and comparisons to other independence notions}
Our results are phrased in terms of {$\ell_2$-independence}. This is closely related to other standard notions in the literature. 

Let $\nu$ be a probability measure on $\{\pm1\}^n$ and write $\Var_\nu(X_i)=1-m_i(\nu)^2$.
Let $\Cov(\nu)$ be the covariance matrix, $\Cor(\nu)$ the correlation matrix, and $\Psi(\nu)$ the influence
matrix. On active coordinates,
\[
\Psi(\nu)=D^{1/2}\Cor(\nu)\,D^{-1/2},
\qquad
D:=\mathrm{diag}(\mathrm{Var}_\nu(X_i)).
\]
Consequently, if the marginals are bounded away from $\pm1$ in the quantitative sense that
\begin{equation}
\label{eq:var-lower-intro}
\Var_\nu(X_i)\ge \sigma^2 \qquad\text{for all active }i,
\end{equation}
then the two operator norms are comparable:
\begin{equation}
\label{eq:cor-vs-psi-intro}
\|\Psi(\nu)\|_{\op}\ \le\ \sigma^{-1}\|\Cor(\nu)\|_{\op}
\qquad\text{and}\qquad
\|\Cor(\nu)\|_{\op}\ \le\ \sigma^{-1}\|\Psi(\nu)\|_{\op}.
\end{equation}
In words, under this assumption, spectral independence, which refers to boundedness of the operator norm (or spectral radius, since these quantities are the same for a symmetric matrix) of the correlation matrix, is equivalent to $\ell_2$-independence up to a factor of $\sigma^{-1}$.

It is well-known that a sufficient condition for $\alpha$-spectral independence is that either the maximum absolute row-sum of $\Psi(\nu)$ (denoted by $\|\Psi(\nu)\|_{\infty\to\infty}$) is bounded by $\alpha$, or the maximum absolute column-sum of $\Psi(\nu)$ (denoted by $\|\Psi(\nu)\|_{1\to1}$) is bounded by $\alpha$. In fact, in many applications, spectral independence is proved by establishing one of these stronger properties. Using the standard matrix inequality
\begin{equation}
\label{eq:l1-linf-to-l2-intro}
\|A\|_{\op}\ \le\ \sqrt{\|A\|_{1\to 1}\,\|A\|_{\infty\to\infty}},
\end{equation}
we observe that in settings where one can control both row and column sums of influences, $\ell_2$-independence follows automatically. For instance, in the case of the hardcore model in the uniqueness regime, it is known that both $\|\Psi\|_{\infty \to \infty}$ (see \cite{chen2023rapid}) and $\|\Psi(\nu)\|_{1 \to 1}$ (see \cite{anari2020spectral}) are bounded by a constant depending only on the relative gap to uniqueness.

\section{Sparse localization and Lipschitzness of the mean map}\label{sec:lipschitz}

This section proves a sparse adaptation of \cite[Lemma~68]{chen2022localization}. We state the proposition below for a general family of directions $\mathcal{V}\subseteq\mathbb{R}^n$; later we specialize to
$\mathcal{V}=\mathrm{Sparse}_c$.

\begin{proposition}
\label{lem:lemma68}
Let $\nu$ be a probability measure on $\{\pm 1\}^n$ and let $\mc{V} \subseteq \mb{R}^n$ with $0 \in \mc{V}$. For $\alpha > 0$, suppose that $\nu$ is $\alpha$ $\ell_2$-independent restricted to $\mathrm{sign}(\mc{V})$. Then, for every $v \in \mathcal{V}$,
\[
\| m(\mathcal{T}_v \nu) - m(\nu) \|_2
\le 4 \alpha \|v\|_2.
\]
\end{proposition}
\begin{remark}
    In \cite[Lemma~68]{chen2022localization}, this statement is proved for the special case $\mc{V} = \mb{R}^n$. We note that arXiv v2 of \cite{chen2022localization} contains a typographical error wherein $\ell_2$-independence is replaced by the weaker assumption of spectral independence; the correct statement with $\ell_2$-independence appears in arXiv v1 of \cite{chen2022localization}. 
\end{remark}

\begin{proof}[Proof of \cref{lem:lemma68}]
Fix $v \in \mathcal{V}$. We will construct a continuous-time jump process $u(t) \in \mathrm{sign}(\mathcal{V})$ (initialized at $u(0) = 0$) and an associated measure-valued process
\[\mu_t := \mathcal{R}_{u(t)} \mathcal{T}_{(1-t)v}\nu\]
such that $(\mu_t)_{t\geq 0}$ is a martingale with respect to the natural filtration corresponding to $\{u(t)\}_{t\geq 0}$. Note that
\[\mu_0 = \mathcal{T}_v \nu, \qquad \mu_1 = \mathcal{R}_{u(1)} \nu.\]
In \cite{chen2022localization}, such a process is constructed for arbitrary $v \in \mathbb{R}^n$, but with $u(t) \in \{-1,0,1\}^n$. We will follow exactly the same construction: our assumption that $v \in \mathcal{V}$ will guarantee that $u(t) \in \mathrm{sign}(\mc{V})$.\\  

\paragraph{\bf Martingale construction}
We describe the construction of $u(t)$. Initialize $u(0) = 0$. For each $i\in[n]$, define
\[
s_i:=\operatorname{sign}(v_i)\in\{-1,0,+1\},
\qquad
a_i:=|v_i|\ge 0,
\]
so that $s_i a_i = v_i$ for all $i \in [n]$. For $u\in \{-1,0,1\}^n$ and $i\in[n]$, define the updated state $u^{(i)} = (u^{(i)}_1,\dots, u^{(i)}_n)$ by
\[
u^{(i)}_j :=
\begin{cases}
s_i & j=i,\\
u_j & j\neq i.
\end{cases}
\]
Given time $t\in[0,1]$ and  $u(t) \in \{-1,0,1\}^n$, define the deterministic measure
\begin{equation}
\label{eq:pi_t_u_def_CE}
\pi_{t,u} := \mathcal{R}_{u(t)}\bigl(\mathcal{T}_{(1-t) v}\nu\bigr).
\end{equation}
Then define the instantaneous jump rate for coordinate $i$ by
\begin{equation}
\label{eq:lambda_def_CE}
\lambda_i(t,u(t))
:= \begin{cases}
a_i\bigl(1+s_i\,m_i(\pi_{t,u})\bigr)
&\qquad \text{for }i\text{ with }u_i=0,\\
0 &\qquad \text{for }i\text{ with }u_i \neq 0
\end{cases}
\end{equation}
Note that, since $m_i(\pi_{t,u})\in[-1,1]$, we have $\lambda_i(t,u(t))\in[0,2a_i]$ and the total rate is bounded by $2\|v\|_1$ uniformly in $(t,u(t))$.

We define $u(t)$ to be the time-inhomogeneous continuous-time Markov chain with generator $\mathcal{L}$ acting on  $F: \{-1,0,1\}^n \to\mathbb{R}$ by
\[
(\mathcal{L}_t F)(u)
= \sum_{i \in [n]} \lambda_i(t,u)\,\bigl(F(u^{(i)})-F(u)\bigr) = 
\sum_{i:\,u_i=0}\lambda_i(t,u)\,\bigl(F(u^{(i)})-F(u)\bigr).
\]
In words, at each time $t$, every coordinate $i$ which is not already $\pm 1$ attempts to jump to $s_i$ with rate $\lambda_i(t, u(t))$; once the coordinate jumps to $s_i$, it stays there forever. In particular, since $v \in \mathcal{V}$, it follows that $u(t) \in \mathrm{sign}(\mc{V})$ for all $t \in [0,1]$.

Let $(\mathcal{F}_t)$ be the natural filtration generated by $u(t)$. We set $\mu_t:=\pi_{t,u(t)}$. The proof that $\mu_t$ is a martingale with respect to $\mc{F}_t$ follows exactly as in \cite[Proposition~18]{chen2022localization}.\\

\paragraph{\bf Lipschitzness}
We couple the martingale $(\mu_t)_{t\geq 0}$ with the measure-valued process $(\nu_t)_{t\geq 0}$ defined by
\[\nu_t = \mathcal{R}_{u(t)}\nu,\]
which is not a martingale. Observe that $\nu_1 = \mu_1$ almost surely. Therefore, since $\mu_t$ is a martingale, we have
\[\mb{E}[m(\nu_1)] = \mb{E}[m(\mu_1)] = \mb{E}[m(\mu_0)] = m(\mu_0),\]
where we have dropped the final expectation since $\mu_0 = \mc{T}_v \nu$ is a deterministic measure.

Then, exactly as in \cite[Eq~66]{chen2022localization}, we have
\begin{align*}
    m(\mathcal{T}_{v} \nu) - m(\nu) &= m(\mu_0) - m(\nu_0)\\
    &= \mb{E}[m(\nu_1)] - m(\nu_0)\\
    &= \int_{0}^{1} \mb{E}[\Psi(\nu_t) P_t Q_t v]dt,
\end{align*}
where $P_t$ and $Q_t$ are diagonal matrices with entries of absolute value at most $2$. In particular,
$\| P_t Q_t v\|_{2} \leq 4 \|v\|_2$. 

Therefore,
\begin{align*}
    \|m(\mathcal{T}_v \nu) - m(\nu)\|_{2}
    &= \bigg\|\int_{0}^{1} \mb{E}[\Psi(\nu_t) P_t Q_t v]dt\bigg\|_{2}\\
    &\leq \int_{0}^{1} \|\mb{E}[\Psi(\nu_t) P_t Q_t v]\|_{2}dt\\
    &\leq \int_{0}^{1} \mb{E} [\|\Psi(\nu_t) P_t Q_t v\|_{2}]dt\\
    &\leq \int_{0}^{1} \mb{E} [\|\Psi(\nu_t)\|_{\mathrm{op}} \|P_t Q_t v\|_{2}]dt\\
    &\leq 4\alpha \|v\|_2,
\end{align*}
where the last line uses that $\nu_t = \mathcal{R}_{u(t)}\nu$, $u(t) \in \mathrm{sign}(\mathcal{V})$, and $\nu$ is $\alpha$-$\ell_2$-independent restricted to $\mathrm{sign}(\mc{V})$. 
\end{proof}

\section{Quadratic stability via sparse convex duality and proof of \cref{thm:entropic-independence}}\label{sec:quad}

\begin{proposition}[Quadratic stability from sparse $\ell_2$-independence]
    \label{prop:quadratic-stability}
Let $\nu$ be a probability measure on $\{\pm 1\}^n$. For $c \in (0,1]$ and $\alpha > 0$, suppose that $\nu$ is $\alpha$ $\ell_2$-independent restricted to $\mathrm{sign}{(\mathrm{Sparse}_c)}$. Then, for all $\mu \ll \nu$,
\[
\|m (\mu) - m(\nu)\|_2^2 \leq \frac{8\alpha}{c}\KL(\mu \|\nu).
\]
\end{proposition}

\begin{remark}
The inequality in \cref{prop:quadratic-stability} is a {quadratic entropic stability} (or
transportation-type) bound: it upper-bounds the squared $\ell_2$-distance between mean vectors by relative
entropy.
When $c=1$, the set $\mathrm{Sparse}_1$ equals $\R^n$, hence $\mathrm{sign}(\mathrm{Sparse}_1)=\{-1,0,1\}^n$,
and \cref{prop:quadratic-stability} recovers the quadratic stability theorem of Chen--Eldan~\cite{chen2022localization}, which assumes a
uniform $\ell_2$-independence bound under {all} pinnings.
The main novelty here is that we only require $\ell_2$-independence under \emph{sparse} pinnings (support at
most $cn$), at the cost of a factor $c^{-1}$. In many applications (including those
motivated by \cite{JMPV}) one can control independence only for a restricted family of pinnings, and the
parameter $c$ quantifies how far such control extends.
\end{remark}

The proof of \cref{prop:quadratic-stability} decomposes into two conceptual steps:
\begin{enumerate}[label=(\roman*),leftmargin=2.2em]
    \item {Sparse localization $\Rightarrow$ Lipschitz mean map.}
    In \cref{lem:lemma68} (proved in the previous section), we adapted the Chen--Eldan negative-fields
    localization argument to show that restricted $\ell_2$-independence over $\mathrm{sign}(\mathcal{V})$
    implies
    \[
    \|m(\mc{T}_v\nu)-m(\nu)\|_2 \le 4\alpha\|v\|_2
    \qquad\forall v\in\mathcal{V}.
    \]
    Applying this with $\mathcal{V}=\mathrm{Sparse}_c$ yields a Lipschitz estimate along sparse directions.
    \item \emph{Sparse Lipschitzness $\Rightarrow$ quadratic stability.}
    We then convert the above Lipschitz estimate into a global quadratic stability inequality
    for all $\mu\ll\nu$ using a restricted convex-duality argument in which we test only {sparse linear} functionals. This introduces a geometric computation of a restricted quadratic
    conjugate over $\mathrm{Sparse}_c$ (captured by Ky Fan norms) and is the reason why the factor
    $c^{-1}$ appears.
\end{enumerate}
We now carry out step~(ii). Step~(i) is provided by \cref{lem:lemma68}.

\subsection{Ky Fan norms and a sparse quadratic conjugate}

\begin{definition}[Ky Fan $m$-norm]\label{def:kyfan}
For $x\in\mathbb{R}^n$ and $m\in\{1,2,\dots,n\}$ define
\[
\|x\|_{2,(m)}^2 \;:=\; \max_{\substack{S\subseteq[n]\\ |S|\le m}}\ \sum_{i\in S} x_i^2.
\]
Equivalently, if $|x|_{(1)}\ge |x|_{(2)}\ge \cdots \ge |x|_{(n)}$ are the absolute values of coordinates of $x$
sorted in nonincreasing order, then
\[
\|x\|_{2,(m)}^2 \;=\; \sum_{j=1}^{m} |x|_{(j)}^2.
\]
\end{definition}

\begin{lemma}[Ky Fan $m$-norm vs.\ $\ell_2$ norm]\label{lem:kyfan_vs_l2}
Let $x\in\mathbb{R}^n$ and $m\in\{1,2,\dots,n\}$. Then
\[
\frac{m}{n}\,\|x\|_2^2 \;\le\; \|x\|_{2,(m)}^2 \;\le\; \|x\|_2^2.
\]
\end{lemma}

\begin{proof} The upper bound is immediate. For the lower bound, write the squared coordinates in nonincreasing order: $a_1\ge a_2\ge \cdots\ge a_n\ge 0$, where $a_j:=|x|_{(j)}^2$. Then \[ \|x\|_{2,(m)}^2=\sum_{j=1}^m a_j. \] Since the average of the top $m$ terms is at least the average of all $n$ terms, \[ \frac{1}{m}\sum_{j=1}^m a_j \;\ge\; \frac{1}{n}\sum_{j=1}^n a_j, \] which rearranges to \[ \sum_{j=1}^m a_j \;\ge\; \frac{m}{n}\sum_{j=1}^n a_j. \] But $\sum_{j=1}^n a_j=\sum_{i=1}^n x_i^2=\|x\|_2^2$, giving \[ \|x\|_{2,(m)}^2 \;\ge\; \frac{m}{n}\,\|x\|_2^2. \qedhere \] \end{proof}

\begin{lemma}[Sparse quadratic conjugate]\label{lem:sparse-conjugate}
Fix $c\in(0,1]$ and $m:=\lceil cn\rceil$. Let $\varepsilon>0$ and define
\[
\Phi_c(x)\;:=\;\sup_{v\in \mathrm{Sparse}_c}\left\{\langle x,v\rangle-\frac{\varepsilon}{2}\|v\|_2^2\right\}.
\]
Then for every $x\in\mathbb{R}^n$,
\[
\Phi_c(x)\;=\;\frac{1}{2\varepsilon}\,\|x\|_{2,(m)}^2
\;\ge\; \frac{c}{2\varepsilon}\,\|x\|_2^2.
\]
\end{lemma}

\begin{proof}
Fix $S\subseteq[n]$ with $|S|\le m$ and restrict to vectors $v$ supported on $S$. Then
\[
\langle x,v\rangle-\frac{\varepsilon}{2}\|v\|_2^2
=\sum_{i\in S}\left(x_iv_i-\frac{\varepsilon}{2}v_i^2\right).
\]
Each summand is maximized at $v_i=x_i/\varepsilon$, yielding
\[
\sup_{\mathrm{supp}(v)\subseteq S}\left\{\langle x,v\rangle-\frac{\varepsilon}{2}\|v\|_2^2\right\}
=\frac{1}{2\varepsilon}\sum_{i\in S}x_i^2.
\]
Maximizing over $S$ gives $\Phi_c(x)=\frac{1}{2\varepsilon}\|x\|_{2,(m)}^2$.
The final inequality follows from \cref{lem:kyfan_vs_l2} and $m/n\ge c$.
\end{proof}

\subsection{From sparse Lipschitzness to quadratic stability}

\begin{proposition}
\label{prop:lemma69}
    Let $\nu$ be a probability measure on $\{\pm 1\}^n$. %with $\on{supp}(\nu) \subseteq \mathrm{SL}_{n,k}$. 
    For $c \in (0,1]$ and $\eps > 0$, suppose that
    \begin{equation}
    \label{eqn:lipschitz-assumption}
        \|m(\mc{T}_v \nu) - m(\nu) \|_{2} \leq \eps \|v\|_2, \qquad \forall v \in \mathrm{Sparse}_c.
    \end{equation}
    Then, for all $\mu \ll \nu$,
    \begin{equation}
    \label{eqn:quadratic-stability}
    \|m(\mu) - m(\nu) \|_{2}^{2} \leq \frac{2\eps}{c} \KL(\mu \| \nu).     \end{equation}
\end{proposition}

\begin{remark}
In \cite[Lemma~69]{chen2022localization}, such a statement is proved under the stronger assumption that \cref{eqn:lipschitz-assumption} holds for all $v \in \mb{R}^n$. We emphasize that while we only assume \cref{eqn:lipschitz-assumption} for $v \in \mathrm{Sparse}_c$, we are able to obtain the conclusion \cref{eqn:quadratic-stability} for all $\mu \ll \nu$, at the expense of an additional factor of $c^{-1}$. 
\end{remark}

The proof of \cref{prop:lemma69} follows the same outline as \cite[Lemma~69]{chen2022localization}. In fact, our proof is a bit simpler, since we observe that one can replace \cite[Lemma~33]{chen2022localization} by the Donsker-Varadhan variational formula (\cref{thm:DV}). 

\begin{proof}
[Proof of \cref{prop:lemma69}]
Define the log-Laplace transform of $\nu$:
\[
f(v)\;:=\;\log\Bigl(\mathbb{E}_{X\sim \nu}\big[e^{\langle v,X\rangle}\big]\Bigr),
\qquad v\in\mathbb{R}^n.
\]
Note that the mean vector $m(\nu) = \mb{E}_{X \sim \nu}[X]$ satisfies
\[m(\nu):=\nabla f(0)\]
and more generally
\[m(\mc{T}_v \nu) = \nabla f(v).\]
Define the centered convex function
\[
h(v)\;:=\; f(v)-\langle m(\nu),v\rangle,
\]
so that $\nabla h(v)=\nabla f(v)-m(\nu)$ and $h(0)=0$.\\

\paragraph{\bf Bound $h(v)$ on sparse vectors}
Fix $v\in\mathrm{Sparse}_c$. Since $\mathrm{Sparse}_c$ is a cone, $tv\in\mathrm{Sparse}_c$ for all $t\in[0,1]$.
By the fundamental theorem of calculus and Cauchy--Schwarz,
\[
h(v)= h(v) - h(0) = \int_0^1 \langle \nabla h(tv),v\rangle\,dt
\le \int_0^1 \|\nabla h(tv)\|_2\,\|v\|_2\,dt.
\]
Using \cref{eqn:lipschitz-assumption} at $tv$ gives $\|\nabla h(tv)\|_2\le \eps \|tv\|_2=\eps t\|v\|_2$, hence
\begin{equation}
h(v)\le \int_0^1 (\eps t\|v\|_2)\|v\|_2\,dt = \frac{\eps}{2}\|v\|_2^2.
\label{eq:h_upper_sparse}
\end{equation}

\paragraph{\bf Apply Donsker-Varadhan with sparse test vectors}
The Donsker--Varadhan variational principle states
\[
\KL(\mu\|\nu)=\sup_{\varphi}\Bigl\{\mathbb{E}_\mu[\varphi]-\log\mathbb{E}_\nu[e^{\varphi}]\Bigr\},
\]
where the supremum ranges over all functions $\varphi$ satisfying $\int e^{\varphi} d\nu < \infty$. Restricting to linear functions $\varphi(\cdot)=\langle v,\cdot\rangle$ with $v\in\mathrm{Sparse}_c$ yields a lower bound:
\[
\KL(\mu\|\nu)
\ge \sup_{v\in\mathrm{Sparse}_c}\Bigl\{\langle m(\mu),v\rangle - f(v)\Bigr\}
= \sup_{v\in\mathrm{Sparse}_c}\Bigl\{\langle m(\mu) - m(\nu),v\rangle - h(v)\Bigr\}.
\]
Using \eqref{eq:h_upper_sparse},
\[
\KL(\mu\|\nu)
\ge \sup_{v\in\mathrm{Sparse}_c}\left\{\langle m(\mu) - m(\nu),v\rangle - \frac{\eps}{2}\|v\|_2^2\right\}.
\]
Now apply \cref{lem:sparse-conjugate} to conclude that
\[
\KL(\mu\|\nu)\ge \frac{c}{2\eps}\|m(\mu) - m(\nu)\|_{2}^2,
\]
as claimed.
\end{proof}

\begin{proof}[Proof of \cref{prop:quadratic-stability} (and hence \cref{thm:main-quadratic-stability})]
Apply \cref{lem:lemma68} with $\mathcal{V}=\mathrm{Sparse}_c$. Since $\nu$ is
$\alpha$-$\ell_2$-independent restricted to $\mathrm{sign}(\mathrm{Sparse}_c)$, we obtain
\[
\|m(\mc{T}_v \nu)-m(\nu)\|_2 \le 4\alpha \|v\|_2
\qquad \forall v\in \mathrm{Sparse}_c.
\]
Thus the hypothesis of \cref{prop:lemma69} holds with $\eps=4\alpha$, and so for every $\mu\ll\nu$,
\[
\|m(\mu)-m(\nu)\|_2^2
\le \frac{2(4\alpha)}{c}\,\KL(\mu\|\nu)
= \frac{8\alpha}{c}\,\KL(\mu\|\nu). \qedhere
\]
\end{proof}

\subsection{Proof of \cref{thm:entropic-independence}}

\begin{proof}[Proof of \cref{thm:entropic-independence}]
Let $\mu\ll\nu$ be measures on $\binom{[n]}{k}$.
Using \cref{fact:kl-chi2-inequality} and the definition of $q_\mu,q_\nu$,
\begin{align*}
\KL(q_{\mu}\|q_{\nu})
&\le \chi^2(q_{\mu}\|q_{\nu})
= \sum_{i\in[n]: q_{\nu}(i) > 0}\frac{\left(\frac{\P_\mu[i\in S]}{k}-\frac{\P_\nu[i\in S]}{k}\right)^2}{\P_\nu[i\in S]/k}
= \frac{1}{k}\sum_{i\in[n]: q_{\nu}(i) > 0}\frac{\left(\P_\mu[i\in S]-\P_\nu[i\in S]\right)^2}{\P_\nu[i\in S]} \\
&\le \frac{1}{bk}\sum_{i\in[n]}\left(\P_\mu[i\in S]-\P_\nu[i\in S]\right)^2,
\end{align*}
where the last inequality uses marginal boundedness.

Let $\widetilde{\mu},\widetilde{\nu}$ be the pushforwards to $\{\pm1\}^n$ as in \cref{subsec:slice-embed}.
Then $\KL(\widetilde{\mu}\|\widetilde{\nu})=\KL(\mu\|\nu)$ and
\[
\|m(\widetilde{\mu})-m(\widetilde{\nu})\|_2^2
=
4\sum_{i=1}^n\left(\P_\mu[i\in S]-\P_\nu[i\in S]\right)^2
\]
by \cref{eq:mean-vs-marg}. By assumption, $\widetilde{\nu}$ is $c$-sparse $\alpha$-$\ell_2$-independent, so
\cref{prop:quadratic-stability} gives
\[
\|m(\widetilde{\mu})-m(\widetilde{\nu})\|_2^2 \le \frac{8\alpha}{c}\,\KL(\widetilde{\mu}\|\widetilde{\nu})
= \frac{8\alpha}{c}\,\KL(\mu\|\nu).
\]
Substituting into the previous inequality yields
\[
\KL(q_{\mu}\|q_{\nu})
\le \frac{1}{bk}\cdot \frac{1}{4}\,\|m(\widetilde{\mu})-m(\widetilde{\nu})\|_2^2
\le \frac{1}{bk}\cdot \frac{1}{4}\cdot \frac{8\alpha}{c}\,\KL(\mu\|\nu)
= \frac{2\alpha}{bck}\,\KL(\mu\|\nu),
\]
i.e.\ $\nu$ is $(2\alpha/(bc))$-entropically independent.
\end{proof}

\section{Application: Down-up walk on fixed-size independent sets}
\label{sec:application}

Let $G=(V,E)$ be a graph with $|V|=n$ and maximum degree at most $\Delta$.
Write $\cI_k(G)$ for the collection of independent sets of size $k$:
\[
\cI_k(G)\;:=\;\{I\subseteq V:\ I\ \text{is independent and}\ |I|=k\},
\]
and let $\mu_k(G)$ be the uniform distribution on $\cI_k(G)$ (when $\cI_k(G)\neq\emptyset$).

Fix $k\ge 1$ and set $\nu:=\mu_k(G)$.
We define a random sequence of measures $\nu_0,\nu_1,\dots,\nu_k$ as follows.
Sample $I\sim \nu$ and, conditional on $I$, sample a uniformly random ordering
$(U_1,\dots,U_k)$ of the elements of $I$ (equivalently, a uniform random permutation of $I$).
For $t\in\{0,1,\dots,k\}$, define the pinned set
\[
S_t \;:=\; \{U_1,\dots,U_t\},
\qquad S_0:=\emptyset,
\]
and define the {residual graph}
\begin{equation}
\label{eq:residual_graph_def}
G_t \;:=\; G\big[V\setminus (S_t\cup N_G(S_t))\big],
\end{equation}
where $N_G(S)$ denotes the neighborhood of $S$ in $G$.
Finally, define $\nu_t$ to be the conditional law of the {remaining} vertices $I\setminus S_t$:
\begin{equation}
\label{eq:nut_def}
\nu_t \;:=\; \mathrm{Law}\big(I\setminus S_t \,\big|\, S_t\big),
\end{equation}
viewed as a measure on $\cI_{k-t}(G_t)$.

\begin{lemma}
\label{lem:uniformity_residual}
For every $t\in\{0,1,\dots,k\}$, conditional on $G_t$ (or on $S_t$), the measure $\nu_t$
is the uniform distribution on $\cI_{k-t}(G_t)$.
\end{lemma}

\begin{proof}
Fix $t$ and condition on $S_t$.
Every $I\in\cI_k(G)$ with $S_t\subseteq I$ can be uniquely written as $I=S_t\cup J$ where
$J$ is an independent set of size $k-t$ in $G_t$ (since $I$ is independent, $J$ cannot use any vertex in
$S_t$ or $N_G(S_t)$, and conversely any such $J$ yields an independent set $S_t\cup J$ in $G$).
Thus the map $I\mapsto I\setminus S_t$ is a bijection between
$\{I\in\cI_k(G): S_t\subseteq I\}$ and $\cI_{k-t}(G_t)$.
Since $\mu_k(G)$ is uniform, its conditioning on $\{S_t\subseteq I\}$ is uniform on that fiber, and hence
the pushforward distribution of $I\setminus S_t$ is uniform on $\cI_{k-t}(G_t)$.
\end{proof}

Recall (following \cite{JMPV}) the constant
\[
\alpha_c(\Delta)
:= \frac{(\Delta-1)^{\Delta-1}}{(\Delta-2)^{\Delta} + (\Delta+1)(\Delta-1)^{\Delta-1}}.
\]
The statement \cite[Eq.~(16)]{JMPV} asserts that for $k=\Omega_\Delta(n)$ and for any $\gamma n \leq k \leq (1-\delta)\alpha_c(\Delta)n$ with
$\ell=\Omega_\Delta(k)$ (there is a typographical error in \cite{JMPV}, where it instead says $k-\ell = \Omega_{\Delta}(k)$,
\begin{equation}
\label{eq:JMPV16_target}
\frac{\E\big[\Ent_{\nu_{k-\ell}}[f]\big]}{\Ent_{\nu_0}[f]}
\;=\;\Omega_{\delta,\Delta, \gamma}(1)
\qquad\text{for all }f:\cI_k(G)\to\R_{\ge 0},
\end{equation}
where the expectation is over the randomness of the subset localization scheme.

The goal of this section is to prove \cref{eq:JMPV16_target}. We will use the following two inputs from \cite{JMPV}:

\begin{lemma}[Bounded marginals for $\mu_r$ {\cite[Lemma~21]{JMPV}}]
\label{lem:JMPV_marginals}
Fix $\Delta\ge 3$ and $\delta,\gamma>0$. There exists $b=b(\Delta,\delta,\gamma)>0$ such that the following
holds.
For every graph $H$ of maximum degree at most $\Delta$ and every integer $r$ satisfying
$\gamma|V(H)|\le r \le (1-\delta)\alpha_c(\Delta)|V(H)|$, we have for all $u\in V(H)$,
\[
\min\{\P_{I\sim \mu_r(H)}[u\in I],\ \P_{I\sim \mu_r(H)}[u\notin I]\}\ \ge\ b.
\]
\end{lemma}

\begin{lemma}[$\ell_\infty$-independence for $\mu_r$ {\cite[Theorem~8]{JMPV}}]
\label{lem:JMPV_linf_indep}
Fix $\Delta\ge 3$ and $\delta,\gamma>0$. There exists $\eta=\eta(\Delta,\delta,\gamma)>0$ such that the
following holds.
For every graph $H$ of maximum degree at most $\Delta$ and every integer $r$ satisfying
$\gamma|V(H)|\le r \le (1-\delta)\alpha_c(\Delta)|V(H)|$, and every $u\in V(H)$,
\[
\sum_{v\in V(H)}
\big|\P_{I\sim \mu_r(H)}[v\in I\mid u\in I]-\P_{I\sim \mu_r(H)}[v\in I\mid u\notin I]\big|
\ \le\ \eta.
\]
Consequently,
\[\|\mathrm{Cor}(\mu_r(H)) \|_{\mathrm{op}} \leq \eta.\]
\end{lemma}

We use the previous two lemmas to deduce sparse $\ell_2$-independence (and consequently, entropic independence) for $\mu_r$.

\begin{lemma}[Sparse $\ell_2$-independence for $\mu_r$ from \cite{JMPV}]
\label{lem:sparse_l2_indep_muk}
Fix $\Delta\ge 3$ and $\delta,\gamma>0$. There exist constants
$c=c(\Delta,\delta,\gamma)\in(0,1)$ and $\alpha=\alpha(\Delta,\delta,\gamma)>0$ such that the following holds.

For every graph $H$ of maximum degree at most $\Delta$ and every integer $r$ satisfying
$\gamma|V(H)|\le r \le (1-\delta)\alpha_c(\Delta)|V(H)|$, the measure $\mu_r(H)$ is $\alpha$ $\ell_2$-independent restricted to $\mathrm{sign}(\mathrm{Sparse}_c)$.
\end{lemma}

\begin{proof}
Let $n:=|V(H)|$ and set $c:=\min\{\gamma/4,\ \delta/8\}$.
Fix any feasible pinning $u\in\{-1,0,1\}^{V(H)}$ with $|\supp(u)|\le cn$; write
$S^+:=\{i:\ u_i=+1\}$ and $S^-:=\{i:\ u_i=-1\}$.
Conditioning $\mu_r(H)$ on the event $\{S^+\subseteq I,\ S^-\cap I=\emptyset\}$ yields the uniform
distribution on independent sets of size $r':=r-|S^+|$ in the graph
\[
H' \;:=\; H\big[V(H)\setminus (S^+\cup N_H(S^+)\cup S^-)\big],
\]
whenever the event has positive probability.
In particular, $H'$ has maximum degree at most $\Delta$ and
\[
n':=|V(H')|
\ \ge\ n-(\Delta+1)|S^+|-|S^-|
\ \ge\ n-(\Delta+2)cn.
\]
Moreover, since $|S^+|\le cn$ and $c\le \gamma/4$, we have
\[
r' \;=\; r-|S^+| \;\ge\; \gamma n - cn \;\ge\; (\gamma/2)n \;\ge\; (\gamma/2)n',
\]
using $n'\le n$. For the upper bound, write $s:=|S^+|$ and $t:=|S^-|$. Since
\[
n' \;\ge\; n-(\Delta+1)s-t,
\]
we have
\[
\frac{r'}{n'}
\;=\;
\frac{r-s}{n'}
\;\le\;
\frac{r-s}{\,n-(\Delta+1)s-t\,}
\;=:\;
g_t(s).
\]
For fixed $t$, a direct computation gives
\[
g_t'(s)
=
\frac{(\Delta+1)r-n+t}{\bigl(n-(\Delta+1)s-t\bigr)^2}.
\]
Now
\[
(\Delta+1)\alpha_c(\Delta)
=
\frac{(\Delta+1)(\Delta-1)^{\Delta-1}}
{(\Delta-2)^{\Delta}+(\Delta+1)(\Delta-1)^{\Delta-1}}
<1,
\]
and also
\[
(\Delta+1)\alpha_c(\Delta)
\;\ge\;
\frac{(\Delta+1)(\Delta-1)^{\Delta-1}}
{(\Delta-1)^{\Delta}+(\Delta+1)(\Delta-1)^{\Delta-1}}
=
\frac{\Delta+1}{2\Delta}
>
\frac12.
\]
Since $r\le (1-\delta)\alpha_c(\Delta)n$ and $t\le cn\le (\delta/8)n$, it follows that
\[
(\Delta+1)r-n+t
\;\le\;
\Bigl((\Delta+1)(1-\delta)\alpha_c(\Delta)-1+\delta/8\Bigr)n
<0.
\]
Hence $g_t$ is decreasing, so its maximum occurs at $s=0$. Therefore
\[
\frac{r'}{n'}
\;\le\;
g_t(0)
=
\frac{r}{n-t}
\;\le\;
\frac{r}{(1-c)n}
\;\le\;
\frac{(1-\delta)\alpha_c(\Delta)}{1-c}
\;\le\;
(1-\delta/2)\alpha_c(\Delta),
\]
where the last inequality uses $c\le \delta/8$.

Therefore, for every such pinning $u$, the conditioned measure $\mc R_u\mu_r(H)$ is exactly
$\mu_{r'}(H')$ for some graph $H'$ with maximum degree at most $\Delta$ and with
\[
(\gamma/2)n' \le r' \le (1-\delta/2)\alpha_c(\Delta)n'.
\]
Applying \cref{lem:JMPV_linf_indep} to $(H',r')$ (with parameters $\gamma/2$ and $\delta/2$) yields
spectral independence of $\mc R_u\mu_r(H)$ with constant $\eta(\Delta,\delta/2,\gamma/2)$. Since $\mu_{r'}(H')$ has marginals lower bounded by $b(\Delta, \delta/2, \gamma/2)$ by \cref{lem:JMPV_marginals}, we have $\on{Var}_{\mu_{r'}(H')} \geq 4b(1-b)$ for all active $i$, and therefore, it follows from \cref{eq:cor-vs-psi-intro} that
\[
\|\Psi(\mc R_u\mu_r(H))\|_{\op}\ \le\ \alpha(\Delta,\delta/2,\gamma/2),
\]
which completes the proof. 
\end{proof}

\begin{proposition}[Entropic independence for $\mu_r$]
\label{prop:EI_muk}
Fix $\Delta\ge 3$ and $\delta,\gamma>0$. There exists $C=C(\Delta,\delta,\gamma)>0$ such that the following holds.
For every graph $H$ of maximum degree at most $\Delta$ and every integer $r$ satisfying
$\gamma|V(H)|\le r \le (1-\delta)\alpha_c(\Delta)|V(H)|$, the measure $\mu_r(H)$ is $C$-entropically independent.
\end{proposition}

\begin{proof}
Combine the marginal bound \cref{lem:JMPV_marginals} (giving $b=b(\Delta,\delta,\gamma)>0$)
with sparse $\ell_2$-independence \cref{lem:sparse_l2_indep_muk} (giving $c,\alpha$),
and apply \cref{thm:entropic-independence}.
\end{proof}

\begin{lemma}[Density window along the localization chain]
\label{lem:residual_density_window}
Fix $d\in(0,1]$ and suppose $\ell\ge dk$. Set $T:=k-\ell$, and for $t\in\{0,1,\dots,T\}$ write
$n_t:=|V(G_t)|$.
Then for every $t\in\{0,1,\dots,T\}$,
\[
d\gamma\, n_t \;\le\; k-t \;\le\; (1-\delta)\alpha_c(\Delta)\, n_t.
\]
Consequently, conditional on $G_t$, the measure $\nu_t=\mu_{k-t}(G_t)$ lies in the parameter regime of
\cref{prop:EI_muk} with lower-density parameter $d\gamma$ and upper-density slack $\delta$.
\end{lemma}

\begin{proof}
The identification $\nu_t=\mu_{k-t}(G_t)$ conditional on $G_t$ is exactly
\cref{lem:uniformity_residual}.

For the lower bound, since $n_t\le n$ and $t\le T=k-\ell$, we have
\[
k-t \;\ge\; \ell \;\ge\; dk \;\ge\; d\gamma n \;\ge\; d\gamma n_t.
\]

For the upper bound, note from \eqref{eq:residual_graph_def} that
\[
n_t \;\ge\; n-(\Delta+1)t.
\]
Therefore
\[
\frac{k-t}{n_t}
\;\le\;
\frac{k-t}{\,n-(\Delta+1)t\,}
\;=:\;
h(t).
\]
A direct computation gives
\[
h'(t)
=
\frac{(\Delta+1)k-n}{\bigl(n-(\Delta+1)t\bigr)^2}.
\]
Since $k\le (1-\delta)\alpha_c(\Delta)n$ and $(\Delta+1)\alpha_c(\Delta)<1$, we have
\[
(\Delta+1)k-n
\;\le\;
\bigl((\Delta+1)(1-\delta)\alpha_c(\Delta)-1\bigr)n
<0.
\]
Hence $h$ is decreasing on $[0,T]$, and so
\[
\frac{k-t}{n_t}
\;\le\;
h(t)
\;\le\;
h(0)
=
\frac{k}{n}
\;\le\;
(1-\delta)\alpha_c(\Delta).
\]
This proves the claim.
\end{proof}

We can now prove the entropy-conservation statement \cref{eq:JMPV16_target}.

\begin{proof}[Proof of \cref{eq:JMPV16_target}]
This is a direct consequence of the ``local-to-global'' entropy  argument from
\cite[Appendix~B]{anari2021entropic}. 
For completeness, we sketch the deduction in our notation.

Choose a constant $d=d(\Delta)>0$ such that $\ell \ge dk$, and write $T:=k-\ell$
(so $\nu_T=\nu_{k-\ell}$).
% Fix $d = \Omega_\Delta(1) $ and $\ell \geq dk$ and write $T:=k-\ell$ (so $\nu_T=\nu_{k-\ell}$).
For each $t\in\{0,1,\dots,T-1\}$, let $\mu_t$ be the $f$-tilt of $\nu_t$, i.e.
$d\mu_t/d\nu_t=f/\mathbb{E}_{\nu_t}f$.
Then $\KL(\mu_t\|\nu_t)=\Ent_{\nu_t}[f]/\mathbb{E}_{\nu_t}f$ by \eqref{eq:EntKL}.
The standard entropy decomposition for subset localization (the same telescoping device used in
\cite[Appendix~B]{anari2021entropic}) gives
\[
\mathbb{E}\!\left[\Ent_{\nu_{t+1}}[f]\mid \nu_t\right]
=
\Ent_{\nu_t}[f]\;-\;\mathbb{E}_{\nu_t}[f]\cdot \KL(p_t\|q_t),
\]
where $q_t$ is the single-vertex distribution induced by $\nu_t$ and $p_t$ is the corresponding
single-vertex distribution induced by $\mu_t$.

By \cref{lem:residual_density_window}, conditional on $G_t$ the realized measure $\nu_t$ is of the form
\[
\nu_t=\mu_{k-t}(G_t)
\]
with
\[
d\gamma\,|V(G_t)| \le k-t \le (1-\delta)\alpha_c(\Delta)\,|V(G_t)|.
\]
Therefore \cref{prop:EI_muk} applies pathwise with lower-density parameter $d\gamma$, and shows that
$\nu_t$ is $C$-entropically independent with
\[
C=C(\Delta,\delta,d\gamma).
\]
Hence
\[
\KL(p_t\|q_t)\;\le\;\frac{C}{k-t}\KL(\mu_t\|\nu_t)
=\frac{C}{k-t}\cdot \frac{\Ent_{\nu_t}[f]}{\mathbb{E}_{\nu_t}f}.
\]
Substituting yields
\[
\mathbb{E}\!\left[\Ent_{\nu_{t+1}}[f]\mid \nu_t\right]
\;\ge\;
\Ent_{\nu_t}[f]\left(1-\frac{C}{k-t}\right).
\]
Taking expectations and iterating from $t=0$ to $t=T-1$ gives
\[
\mathbb{E}[\Ent_{\nu_T}[f]]
\;\ge\;
\left(\prod_{j=\ell+1}^{k}\Bigl(1-\frac{C}{j}\Bigr)\right)\Ent_\nu[f].
\]
Since $\ell\geq dk$ and $C=C(\Delta,\delta,d\gamma)$ is independent of $n$, for all sufficiently
large $n$ we have $C/j\le 1/2$ for every $j\in\{\ell+1,\dots,k\}$; the finitely many remaining
values of $n$ can be absorbed into the final constant. Therefore
\[
\prod_{j=\ell+1}^{k}\Bigl(1-\frac{C}{j}\Bigr)
\;\ge\;
\exp\!\left(-2C\sum_{j=\ell+1}^{k}\frac1j\right)
\;\ge\;
\exp\!\bigl(-2C\log(k/\ell)\bigr)
\;\ge\;
\exp\!\bigl(-2C\log(1/d)\bigr)
\;=:\;
c_0(\Delta,\delta,\gamma)
\;>\;0.
\]
This completes the proof.

\end{proof}

\bibliographystyle{abbrv}
\bibliography{main.bib}

@inproceedings{JMPV,
  title={Optimal mixing of the down-up walk on independent sets of a given size},
  author={Jain, Vishesh and Michelen, Marcus and Pham, Huy Tuan and Vuong, Thuy-Duong},
  booktitle={2023 IEEE 64th Annual Symposium on Foundations of Computer Science (FOCS)},
  pages={1665--1681},
  year={2023},
  organization={IEEE}
}

@article{Chen2022OptimalMF,
  title={Optimal mixing for two-state anti-ferromagnetic spin systems},
  author={Xiaoyu Chen and Weiming Feng and Yitong Yin and Xinyuan Zhang},
  journal={2022 IEEE 63rd Annual Symposium on Foundations of Computer Science (FOCS)},
  year={2022},
  pages={588-599},
  url={https://api.semanticscholar.org/CorpusID:247450653}
}

@inproceedings{chen2022localization,
  title={Localization schemes: A framework for proving mixing bounds for {M}arkov chains},
  author={Chen, Yuansi and Eldan, Ronen},
  booktitle={2022 IEEE 63rd Annual Symposium on Foundations of Computer Science (FOCS)},
  pages={110--122},
  year={2022},
  organization={IEEE}
}

@article{chen2023rapid,
  title={Rapid mixing of Glauber dynamics up to uniqueness via contraction},
  author={Chen, Zongchen and Liu, Kuikui and Vigoda, Eric},
  journal={SIAM Journal on Computing},
  volume={52},
  number={1},
  pages={196--237},
  year={2023},
  publisher={SIAM}
}

@inproceedings{anari2020spectral,
  title={Spectral independence in high-dimensional expanders and applications to the hardcore model},
  author={Anari, Nima and Liu, Kuikui and Gharan, Shayan Oveis},
  booktitle={2020 IEEE 61st Annual Symposium on Foundations of Computer Science (FOCS)},
  pages={1319--1330},
  year={2020},
  organization={IEEE}
}

@article{anari2021entropic,
  title={Entropic independence {I}: Modified log-Sobolev inequalities for fractionally log-concave distributions and high-temperature ising models},
  author={Anari, Nima and Jain, Vishesh and Koehler, Frederic and Pham, Huy Tuan and Vuong, Thuy-Duong},
  journal={arXiv preprint arXiv:2106.04105},
  year={2021}
}

@article{anari2021entropic2,
  title={Entropic independence {II}: optimal sampling and concentration via restricted modified log-Sobolev inequalities},
  author={Anari, Nima and Jain, Vishesh and Koehler, Frederic and Pham, Huy Tuan and Vuong, Thuy-Duong},
  journal={arXiv preprint arXiv:2111.03247},
  year={2021}
}

@inproceedings{ChenLiuVigoda2021,
  title={Optimal mixing of {G}lauber dynamics: Entropy factorization via high-dimensional expansion},
  author={Chen, Zongchen and Liu, Kuikui and Vigoda, Eric},
  booktitle={Proceedings of the 53rd Annual ACM SIGACT Symposium on Theory of Computing (STOC)},
  pages={1537--1550},
  year={2021}
}

@article{kuchukova2025fast,
  title={Fast and Slow Mixing of the {K}awasaki Dynamics on Bounded-Degree Graphs},
  author={Kuchukova, Aiya and Pappik, Marcus and Perkins, Will and Yap, Corrine},
  journal={Random Structures \& Algorithms},
  volume={67},
  number={4},
  pages={e70038},
  year={2025},
  publisher={Wiley Online Library}
}
\end{document}